\documentclass[12pt,a4paper]{article}
\usepackage[latin1]{inputenc}
\usepackage{amsmath}
\usepackage{amsfonts}
\usepackage{float}
\usepackage{amssymb}
\usepackage{amsthm}
\usepackage{enumerate}
\newtheorem{defi}{Definition}
\newtheorem{teo}{Theorem}
\newtheorem{lemi}{Lemma}
\newtheorem{propi}{Proposition}

\usepackage{pgf,tikz}\usepackage{mathrsfs}\usetikzlibrary{arrows}
\usepackage{graphicx}
\usepackage{authblk}
\usepackage{rotating}
\usepackage{graphicx}
\usepackage{subcaption}
\usepackage{float}
\begin{document}
	
\title{Effort of rugby teams according to the bonus point system: a theoretical and empirical analysis\thanks { %The authors declare that they have no conflict of interest.
		We thank the support of the Agencia Nacional de Promoci\'on Cient\'{\i}fica y Tecnol\'ogica (Argentina) through grant PICT-2017-2355.  Alejandro Neme also acknowledges  the financial support received from the UNSL (grant PROICO-32016) and from the Consejo Nacional de Investigaciones Cient\'{\i}ficas y T\'ecnicas (Argentina), through grant PIP 112-200801-00655.}}
\author[1]{Federico Fioravanti \footnote{federico.fioravanti9@gmail.com (Corresponding author)}}
\author[2]{Alejandro Neme\footnote{aneme@unsl.edu.ar}}
\author[1]{Fernando Tohm\'e \footnote{ftohme@criba.edu.ar}}
\author[1]{Fernando Delbianco \footnote{fernando.delbianco@uns.edu.ar}} 

\affil[1]{INMABB, Universidad Nacional del Sur, Bah\'{\i}a Blanca, Argentina}
\affil[2]{IMASL, Universidad Nacional de San Luis, San Luis, Argentina}
\date{}
\maketitle
\begin{abstract}
Using a simple game-theoretical model of contests, we compare the effort exerted by rugby teams under three different point systems  used in tournaments around the world. The scoring systems under consideration are {\em NB}, {\em +4} and {\em 3+}.  We state models of the games under the three point systems, both static and dynamic. In all those models we find that the $3+$ system ranks first, $+4$ second and $NB$ third. We run empirical analyses using data from matches under the three scoring systems. The results of those statistical analyses confirm our theoretical conclusions.

%{\em NB} awards four points to the winning team, no points to the losing team, and two points to each team in the case of a tie. {\em +4} and {\em 3+} are like {\em NB} but grant also an extra point to a team that scores {\em four} or {\em three} tries more than the opponent, respectively, while giving one point to teams that lose by only one try.\\
%In a static model we show that teams become more offensive when an extra point is awarded for scoring more than four tries. We also show that not giving extra points makes the team more offensive than giving it for scoring three or more tries than the opponent.\\ 
%On the other hand, using a dynamic model based on Mass\'o - Neme (1996), we compare the sets of feasible and equilibrium payoffs. We find in this case that the system that gives an extra point for scoring four or more tries leads to a better set of feasible and equilibrium payoffs. Unlike the static model, the dynamic model shows that it is preferable to award an extra point to winners for scoring three more tries than their opponents rather than not granting an extra point.
 
\end{abstract}

\section {Motivation}
Rugby is a sport in constant evolution.  This can be easily seen in the fact that its rules are being continuously reevaluated. This leads to experiments in which alternative rules are tested and, if the results are satisfactory, become part of the rulebook.\\
These continuous revisions and modifications are intended to increase both the safety of the players and the pleasure of watching the game.  That is, some rules are changed to make the game safer, reducing the number of injuries suffered by the players, while others are modified to make the game more entertaining for players, coaches and spectators alike.\\
Some of the latter kind of modifications involve the points awarded to teams, depending on, among other factors, the number of tries scored. Fortunately for the exploration of alternative rules, the organizers of tournaments are given a free hand to choose the point system to be applied and thus to experiment with those variations. For example, in the World Cup (or in the Rugby Championship, a tournament in which the only participants are the best four national teams of the Southern Hemisphere), four points are awarded to the winning team, two to each team in case of a tie and no points to the losing team. 
Besides the points awarded for winning, tying or losing a game, an extra point (usually called bonus point) is awarded to the team that scores four or more tries, and an extra point to the losing team, if the difference in the score is seven points or less.\\ 
Another point system used around the world, for example in the Super Rugby or French Top 14, two of the most important club tournaments in the world, consists in giving an extra point to the winning team if it scores three more tries than its rival, and an extra point is awarded to the losing team if the difference in the score is seven or less.\\ 
In this work we compare the effort of the teams under these two point systems as well as under the point system in which no bonus points are awarded. Clearly, there is no consensus on which one is the best, reflected in the fact that different important tournaments around the world use different point systems. But there exists a consensus on that games are more entertaining  when teams fight to the end to win a match. This behavior can be induced by point rules that give incentives to teams to exert more effort in order to succeed.\\ 
To start analyzing the effort aspects of point systems, notice that there are different ways in which a team can score in rugby . One is by grounding the ball in the other team's ingoal.\footnote{An ingoal is a rectangular area at the end of the field, which has two of them, each one corresponding to one of the teams.} This is called {\em scoring a try} and the team that does it, earns five points. When a team scores a try, it can opt to kick to the posts, getting two more points. On the other hand, a kick from the ground (called a penalty kick) or a drop kick (a kick after the ball bounced) that goes through the posts, awards the team three points.\\
The aforementioned scoring methods are under evaluation. World Rugby (the international governing body of rugby union) is interested in knowing whether giving more points for a try or giving less for a place kick induces teams to score more tries. \\
A common understanding of a team devoting more effort in a game is that it plays more offensively, defends with more attitude and forces its players to play as hard as they can. We want to see if the effort spent in each bonus point system correlates in some way with the number of tries. We intend to determine, resorting to a game-theoretical analysis, which point system induces teams to devote more effort. Furthermore, using real-world data we check the validity of our theoretical conclusions, which can be useful for sports planners who intend to design tournaments on sound theoretical and empirical grounds.\\
Although a sport like rugby can be really difficult to model, due to the presence of variables with an uncountable number of possible values, there are effective ways of simplifying the analysis. The use of a simple game-theoretical model of contests allows to predict the behavior of the teams and to find the most adequate strategies for each instance of a match.\\ 
Some of the authors that have modeled different aspects of sports using game-theoretical tools are Walker and Woodens (2001) for tennis, Chiappori, Levitt and Groseclose (2002), Palomino, Rigotti and Rustichini (1999) in the case of soccer and Petr\'{o}czi and Haugen (2012) to evaluate the effectiveness of anti-doping policies. An analysis particularly relevant for our purposes compares the strategies of two soccer teams under the two and the three points scoring systems (Brocas and Carrillo, 2004). The authors conclude, rather unsurprisingly, that teams become more offensive if they are awarded three points when they win. But interestingly, they also find that by giving more than three points to a winner makes the teams  more defensive in the first half of the game and so, in average, higher offensiveness is not induced by this point system.\\
In our case, we start modeling a rugby game statically, comparing the effort devoted by the teams under different point systems. We seek to find out which one pushes teams to the limit, making the game more entertaining.  We then try to answer the same question, this time in the framework of a dynamic model, using the results of Mass\'o - Neme (1996), by analyzing the feasible and equilibrium payoffs. We consider the average of joint efforts used to obtain those equilibrum payoffs in each point system, to find out which one induces the teams to play more agressively. Finally, we check the real-world validity of our conclusions using data from different tournaments around the world.\\
The plan of the paper is as follows. In Section 2 we present the static model and examine the degrees of offensiveness associated to the different point systems.  In Section 3 we do the same, but in the context of a dynamic model. We find that the order of offensiveness is the same in both cases, being $3+$ the system that ranks on the top.  In Section 4  we run empirical analyses in order to corroborate the validity of those results, which are confirmed by the data of various tournaments under the three systems. Finally, Section 5 concludes.

\section{The Static Model}
We intend here to model, in a simple way, the effects of point systems on the choice of the levels of effort of teams. We consider two teams, $A$ and $B$. The possible events in a match are denoted $(a,b)\in \mathbb{N}_{0}\times\mathbb{N}_{0}$, where $\mathbb{N}_{0}$ represents the natural numbers plus $0$. Letters $a$ and $b$ stand for the tries scored by teams $A$ and $B$, respectively.\\ 
To simplify the analysis, we disregard the precise differences between goals, penalty kicks or drop kicks, and just focus on the tries scored and the joint efforts of the teams. In each event, we consider a contest in which two risk - neutral contestants are competing to score a try, and win the points awarded by the points system.\footnote{We consider that teams totally discount the future, and assume that the game can end after they score a try.} The contestants differ in their valuation of the prize. Each contestant $i\in\{A,B\}$ independently exerts an irreversible and costly effort $e_{i}\geq0$, which will determine, through a {\it contest success function} (CSF), which team wins the points. Formally, the CSF maps the profile of efforts $(e_{A},e_{B})$ into probabilities of scoring a try. We adopt the logit formulation, since it is the most widely used in the analysis of sporting contests (Dietl et al., 2011). Its general form was introduced by Tullock (1980), although we use it here with a slight modification:\footnote{When teams exert no effort, the probability of scoring a try is 0. This allows to get a tie as a result.}

 $$p_{i}(e_{A},e_{B})= \left\{ \begin{array}{lcc}
             \dfrac{e_{i}^{\alpha}}{e_{A}^{\alpha}+e_{B}^{\alpha}} &   if  & max\{e_{A},e_{B}\}>0 
             \\0 &   & otherwise
             \end{array}
   \right.$$

The parameter $\alpha>0$ is called the \textquotedblleft discriminatory power\textquotedblright{} of the CSF, measuring the sensitivity of success to the level of effort exerted.\footnote{This CSF satisfies homogeneity. That is, when teams exert the same level of effort, they have the same probabilities of winning the contest. This is a plausible hypothesis when teams have the same level of play.} We normalize it and set $\alpha=1$. Associated to effort there is a cost function $c_{i}(e_{i})$, often assumed linear in the literature,

$$c_{i}(e_{i})=ce_{i}$$

\noindent where $c>0$ is the (constant) marginal cost of effort.\\

The utility or payoff function when the profile of efforts is $(e_{A},e_{B})$ and the score is $(a,b)$, has the following form (we omit the effort argument for simplicity):

$$U_{A}((e_{A},e_{B}),(a,b))=$$
$$p_{A}(f_{A}(a+1,b)+k_{B1}\epsilon)+(1-p_{A}-p_{B})(f_{A}(a,b)+k_{B2}\epsilon))+p_{B}f_{A}(a,b+1)-ce_{A}$$

\noindent for team $A$, and

$$U_{B}((e_{A},e_{B}),(a,b))=$$
$$p_{A}f_{B}(a+1,b)+(1-p_{A}-p_{B})(f_{B}(a,b)+k_{A2}\epsilon))+p_{B}(f_{B}(a,b+1)+k_{A1})\epsilon)-ce_{B}$$

\noindent for team $B$, where

$$k_{B1}=f_{B}(a,b+1)-f_{B}(a+1,b)$$
$$k_{B2}=f_{B}(a,b+1)-f_{B}(a,b)$$
$$k_{A1}=f_{A}(a+1,b)-f_{A}(a,b+1)$$
$$k_{A2}=f_{A}(a+1,b)-f_{A}(a,b)$$

\noindent and
$f_{i}:\mathbb{N}_{0}\times\mathbb{N}_{0}\rightarrow \{0,1,2,3,4,5\}$  depends on the point system we are working with. It is defined on the final scores and yields the points earned by team $i$. 
Each point system is characterized by a different function.\\

In the case where no bonus point ($NB$ system) is awarded, we have:\\

 $f_{A}^{NB}(a,b)= \left\{ \begin{array}{lcc}
             4 &   if  & a>b 
             \\ 2 &  if & a=b
             \\ 0 &  if  & a<b
             \end{array}
   \right.$

   $f_{B}^{NB}(a,b)= \left\{ \begin{array}{lcc}
                0 &   if  & a>b 
                \\ 2 &  if & a=b
                \\ 4 &  if  & a<b
                \end{array}
      \right.$\\
     
When a bonus point is given for scoring $4$ or more tries, and for losing by one try ($+4$ system) the functions are:\\

$f_{A}^{+4}(a,b)= \left\{ \begin{array}{lcccccccc}
             4 &   if  & a>b & and & a<4 
             \\5 &   if  & a>b & and & a\geq 4 
             \\  2 &  if & a=b & or & [a>4 & and & b-a=1] \\
             3 &   if  & a=b & and & a\geq 4\\
             0 &   if  & a<4 & and & b-a>1\\
             1 & if & [a+1<b & and & a\geq 4] & or & b-a=1
             \end{array}
   \right.$

$f_{B}^{+4}(a,b)= \left\{ \begin{array}{lcccccccc}
                4 &   if  & b>a & and & b<4 
                \\5 &   if  & b>a & and & b\geq 4 
                \\  2 &  if & a=b & or & [b>4 & and & a-b=1] \\
                3 &   if  & a=b & and & b\geq 4\\
                0 &   if  &b<4 & and & a-b>1\\
                1 & if & [b+1<a & and & b\geq 4] & or &  a-b=1
                \end{array}
      \right.$\\
      
Finally, when a difference of $3$ tries gives the winning team a bonus point and losing by one try gives the bonus point to the loser ($3+$ system) we have:\\

$f_{A}^{3+}(a,b)= \left\{ \begin{array}{lcccccccccc}
             4 &   if  &  0<a-b<3 
             \\5 &   if  &a-b\geq 3 
             \\  2 &  if & a=b  \\
             0 &   if  & b-a>1 \\
             1 & if & b-a=1
             \end{array}
   \right.$

      $f_{B}^{3+}(a,b)= \left\{ \begin{array}{lcccccccccc}
                   4 &   if  & 0<b-a<3 
                   \\5 &   if & b-a\geq 3 
                   \\  2 &  if & a=b  \\
                   0 &   if & a-b>1 \\
                   1 & if & a-b=1
                   \end{array}
         \right.$    \\  
         
In all three cases the utility functions represent the weighted sum of three probabilities, namely that of team $A$ scoring, that of none of the teams scoring and that of team $B$ scoring. The corresponding weights are the points earned in each case plus the gain of blocking the other team, precluding it of winning points. This gain is defined as the difference between the points that the other team can earn if it scores and the points they get times $\epsilon$, where $0<\epsilon<<1$ is not very large. This $\epsilon$ intends to measure the importance of blocking the other team and not letting it score  and earn more points. Teams are playing a tournament, so making it hard for the other team to earn points is an incentive (although not a great one) in a match. The way we define the utility function rests on the simple idea that to score four tries, one has to be scored first. This captures the assumption that teams care only about the immediate result of scoring, and not about what can happen later.\\ 
Under these assumptions, we seek to find the equilibria corresponding to the three point systems. The appropriate notion of equilibrium here is  in terms of  {\em strict dominant strategies} since the chances of each team are independent of what the other does. Notice that, trivially, each dominant strategies equilibrium is (the unique) Nash equilibrium in the game.\footnote{Nash equilibria exist since the game trivially satisfies the condition of having  compact and convex spaces of strategies while the utility functions have the expected probability form, which ensures that the best response correspondence has a fixed point.} Once obtained these equilibria, the next step of the analysis is to compare them, to determine how the degree of offensiveness changes with the change of rules. This comparison is defined in terms of the following relation:

$$(e_{A},e_{B})\succeq (e_{A^{\prime}},e_{B^{\prime}})\ \mbox{if}\ e_{A}+e_{B}\geq e_{A^{\prime}}+e_{B^{\prime}}$$
\noindent while
$$(e_{A},e_{B})\sim (e_{A^{\prime}},e_{B^{\prime}}) \ \ \mbox{in any other case.}$$

\noindent where $(e_{A},e_{B})\succeq (e_{A^{\prime}},e_{B^{\prime}})$ is understood as ``with $(e_{A},e_{B})$ both teams exert more effort than with $(e_{A^{\prime}},e_{B^{\prime}})$''.\\

We look for the maximum number of tries that can be scored by a team, in order to limit the number of cases to analyze. We use the statistics of games played in different tournaments around the world, which show that, in average, teams can get at most $7$ tries ([12]-[23])(see Section 4). This, in turn leads us to $64$ possible instances ({\em events}).\\

At each event we compare the equilibrium strategies. We thus obtain a ranking of the point systems, based on the $\succeq$ relation. The reaction function of team $i$, describing the best response to any possible effort choice of the other team, can be computed from the following first order conditions:
$$\dfrac{e_{B}}{(e_{A}+e_{B})^{2}}k_{A}=c $$

\noindent for team $A$, and

$$\dfrac{e_{A}}{(e_{A}+e_{B})^{2}}k_{B}=c  $$

\noindent for team $B$, where $k_{A}$ and $k_{B}$ obtain by rearranging the constants of the corresponding utility function. The equilibrium $(e_{A}^{\ast},e_{B}^{\ast})$ in pure strategies is characterized by the intersection of the two reaction functions and is given by:

$$(e_{A}^{\ast},e_{B}^{\ast})=(\dfrac{k_{A}^{2}k_{B}}{c(k_{A}+k_{B})^{2}},\dfrac{k_{A}k_{B}^{2}}{c(k_{A}+k_{B})^{2}}) $$

As an example, consider, without loss of generality, a particular instance (for simplicity we omit the arguments):

\begin{itemize}
\item \textbf{Event (2,0)}\\
\textbf{$NB$ system}
$$U_{A}((e_{A},e_{B}),(2,0))=p_{A}(4+0\epsilon)+(1-p_{A}-p_{B})(4+0\epsilon))+p_{B}4-ce_{A}$$
$$U_{B}((e_{A},e_{B}),(2,0))=p_{A}0+(1-p_{A}-p_{B})(0+0\epsilon))+p_{B}(0+0\epsilon)-ce_{B}$$
The Nash (dominant strategies) equilibrium is given by $(e_{A}^{\ast},e_{B}^{\ast})=(0,0) $\\

\textbf{$3+$ system}
$$U_{A}((e_{A},e_{B}),(2,0))=p_{A}(5+\epsilon)+(1-p_{A}-p_{B})(4+\epsilon))+p_{B}4-ce_{A}$$
$$U_{B}((e_{A},e_{B}),(2,0))=p_{A}0+(1-p_{A}-p_{B})(0+\epsilon))+p_{B}(1+\epsilon)-ce_{B}$$
The equilibrium is $(e_{A}^{\ast},e_{B}^{\ast})=(\dfrac{(1+\epsilon)^{2}(1+\epsilon)}{c(1+\epsilon+1+\epsilon)^{2}},\dfrac{(1+\epsilon)(1+\epsilon)^{2}}{c(1+\epsilon+1+\epsilon))^{2}}) $.\\

\textbf{$+4$ system}
$$U_{A}((e_{A},e_{B}),(2,0))=p_{A}(4+\epsilon)+(1-p_{A}-p_{B})(4+\epsilon))+p_{B}4-ce_{A}$$
$$U_{B}((e_{A},e_{B}),(2,0))=p_{A}0+(1-p_{A}-p_{B})(0+0\epsilon))+p_{B}(1+0\epsilon)-ce_{B}$$
The equilibrium is $(e_{A}^{\ast},e_{B}^{\ast})=(\dfrac{\epsilon^{2}1}{c(1+\epsilon)^{2}},\dfrac{\epsilon1^{2}}{c(1+\epsilon)^{2}}) $.\\
\end{itemize}

Since we assume that $\epsilon$ is sufficiently small, we can infer that teams will exert more effort under the $3+$ system, then in the $+4$ and finally in the $NB$. The comparison of all the possible events yields:

\begin{propi}
The $3+$ system is the Condorcet winner in the comparison among the point systems. By the same token, teams exert more effort under the $+4$ system than in the $NB$ one.
\end{propi}
\begin{proof}
We analyze the 64 possible events. Table $1$ shows the results favoring team $A$. By symmetry, analogous results can be found for team $B$. \\
Consider the following pairwise comparisons:
\begin{itemize}
	\item $NB$ vs. $3+$: 22 events rank higher under {\em 3+}, while 7 under {\em NB}.
	\item $+4$ vs. $NB$: 22 events for the former against 7 with the latter.
	\item $3+$ vs. $+4$: 16 with the former against 6 with {\em +4}.
\end{itemize}
This indicates that $3+$ is the Condorcet winner, while $NB$ is the Condorcet loser. 
\end{proof}

\begin{table}[H]~\label{table}
	\begin{flushleft}
		\resizebox{17cm}{9.5cm}{
			\begin{tabular}{||l l l l l||} 
				\hline
				Events & {\em NB} Equilibrium & {\em 3+} Equilibrium & {\em +4} Equilibrium & Ranking \\ 
				\hline
				(0,0),(1,1),(2,2) & $ (\dfrac{(4+4\epsilon)^{3}}{4c(4+4\epsilon)^{2}},\dfrac{(4+4\epsilon)^{3}}{4c(4+4\epsilon)^{2}})$ & $(\dfrac{(3+3\epsilon)^{3}}{4c(3+3\epsilon)^{2}},\dfrac{(3+3\epsilon)^{3}}{4c(3+3\epsilon)^{2}})$&$(\dfrac{(3+3\epsilon)^{3}}{4c(3+3\epsilon)^{2}},\dfrac{(3+3\epsilon)^{3}}{4c(3+3\epsilon)^{2}})$& $NB\succ3+\sim +4$ \\ 
				
				(3,3)  & $(\dfrac{(4+4\epsilon)^{3}}{4c(4+4\epsilon)^{2}},\dfrac{(4+4\epsilon)^{3}}{4c(4+4\epsilon)^{2}})$ & $(\dfrac{(3+3\epsilon)^{3}}{4c(3+3\epsilon)^{2}},\dfrac{(3+3\epsilon)^{3}}{4c(3+3\epsilon)^{2}})$&$(\dfrac{(4+4\epsilon)^{3}}{4c(4+4\epsilon)^{2}},\dfrac{(4+4\epsilon)^{3}}{4c(4+4\epsilon)^{2}})$& $NB\sim+4\succ 3+$ \\ 
				
				(4,4),(5,5),(6,6), (7,7)& $ (\dfrac{(4+4\epsilon)^{3}}{4c(4+4\epsilon)^{2}},\dfrac{(4+4\epsilon)^{3}}{4c(4+4\epsilon)^{2}})$ & $(\dfrac{(3+3\epsilon)^{3}}{4c(3+3\epsilon)^{2}},\dfrac{(3+3\epsilon)^{3}}{4c(3+3\epsilon)^{2}})$&$(\dfrac{(3+3\epsilon)^{3}}{4c(3+3\epsilon)^{2}},\dfrac{(3+3\epsilon)^{3}}{4c(3+3\epsilon)^{2}})$& $NB\succ3+\sim +4$ \\ 
				
				(1,0),(2,1)  & $ (\dfrac{(2+2\epsilon)^{3}}{4c(2+2\epsilon)^{2}},\dfrac{(2+2\epsilon)^{3}}{4c(2+2\epsilon)^{2}})$ & $(\dfrac{(2+2\epsilon)^{3}}{4c(2+2\epsilon)^{2}},\dfrac{(2+2\epsilon)^{3}}{4c(2+3\epsilon)^{2}})$&$(\dfrac{(2+2\epsilon)^{3}}{4c(2+2\epsilon)^{2}},\dfrac{(2+2\epsilon)^{3}}{4c(2+2\epsilon)^{2}})$& $NB\sim3+\sim +4$ \\ 
				(3,2) & $(\dfrac{(2+2\epsilon)^{3}}{4c(2+2\epsilon)^{2}},\dfrac{(2+2\epsilon)^{3}}{4c(2+2\epsilon)^{2}})$ &$(\dfrac{(2+2\epsilon)^{3}}{4c(2+2\epsilon)^{2}},\dfrac{(2+2\epsilon)^{3}}{4c(2+3\epsilon)^{2}})$&$(\dfrac{(3+2\epsilon)^{2}(2+3\epsilon)}{c(5+5\epsilon)^{2}},\dfrac{(3+2\epsilon)(2+3\epsilon)^{3}}{c(5+5\epsilon)^{2}})$& $+4\succ3+\sim NB$ \\ 
				(4,3) & $(\dfrac{(2+2\epsilon)^{3}}{4c(2+2\epsilon)^{2}},\dfrac{(2+2\epsilon)^{3}}{4c(2+2\epsilon)^{2}})$ &$(\dfrac{(2+2\epsilon)^{3}}{4c(2+2\epsilon)^{2}},\dfrac{(2+2\epsilon)^{3}}{4c(2+3\epsilon)^{2}})$&$(\dfrac{(2+3\epsilon)^{2}(3+2\epsilon)}{c(5+5\epsilon)^{2}},\dfrac{(2+3\epsilon)(3+2\epsilon)^{3}}{c(5+5\epsilon)^{2}})$& $+4\succ3+\sim NB$ \\
				
				(5,4), (6,5), (7,6)   & $ (\dfrac{(2+2\epsilon)^{3}}{4c(2+2\epsilon)^{2}},\dfrac{(2+2\epsilon)^{3}}{4c(2+2\epsilon)^{2}})$ & $(\dfrac{(2+2\epsilon)^{3}}{4c(2+2\epsilon)^{2}},\dfrac{(2+2\epsilon)^{3}}{4c(2+3\epsilon)^{2}})$&$(\dfrac{(2+2\epsilon)^{3}}{4c(2+2\epsilon)^{2}},\dfrac{(2+2\epsilon)^{3}}{4c(2+2\epsilon)^{2}})$& $NB\sim3+\sim +4$ \\ 
				
				(2,0)&$(0,0)$&$(\dfrac{(1+\epsilon)^{3}}{4c(1+1\epsilon)^{2}},\dfrac{(1+\epsilon)^{3}}{4c(1+\epsilon)^{2}})$&$(\dfrac{\epsilon^{2}}{c(1+\epsilon)^{2}},\dfrac{\epsilon}{c(1+\epsilon)^{2}})$& $3+\succ+4\succ NB$ \\ 
				
				(3,1)&$(0,0)$&$(\dfrac{(1+\epsilon)^{3}}{4c(1+1\epsilon)^{2}},\dfrac{(1+\epsilon)^{3}}{4c(1+\epsilon)^{2}})$&$(\dfrac{(1+\epsilon)^{3}}{4c(1+1\epsilon)^{2}},\dfrac{(1+\epsilon)^{3}}{4c(1+\epsilon)^{2}})$& $3+\sim+4\succ NB$\\
				
				(4,2) &$(0,0)$&$(\dfrac{(1+\epsilon)^{3}}{4c(1+1\epsilon)^{2}},\dfrac{(1+\epsilon)^{3}}{4c(1+\epsilon)^{2}})$&$(\dfrac{\epsilon^{2}}{c(1+\epsilon)^{2}},\dfrac{\epsilon}{c(1+\epsilon)^{2}})$& $3+\succ+4\succ NB$ \\ 
				
				(5,3) &$(0,0)$&$(\dfrac{(1+\epsilon)^{3}}{4c(1+1\epsilon)^{2}},\dfrac{(1+\epsilon)^{3}}{4c(1+\epsilon)^{2}})$&$(\dfrac{8\epsilon^{2}}{4c(1+\epsilon)^{2}},\dfrac{8\epsilon}{4c(1+\epsilon)^{2}})$& $3+\succ+4\succ NB$ \\ 
				(6,4), (7,5) &$(0,0)$&$(\dfrac{(1+\epsilon)^{3}}{4c(1+1\epsilon)^{2}},\dfrac{(1+\epsilon)^{3}}{4c(1+\epsilon)^{2}})$&$(\dfrac{\epsilon^{2}}{c(1+\epsilon)^{2}},\dfrac{\epsilon}{c(1+\epsilon)^{2}})$& $3+\succ+4\succ NB$ \\
				
				(3,0)&$(0,0)$&$(\dfrac{\epsilon}{c(1+\epsilon)^{2}},\dfrac{\epsilon^{2}}{c(1+\epsilon)^{2}})$&$(\dfrac{\epsilon}{c(1+\epsilon)^{2}},\dfrac{\epsilon^{2}}{c(1+\epsilon)^{2}})$& $3+\sim+4\succ NB$ \\
				
				(4,1), (5,2) &$(0,0)$&$(\dfrac{\epsilon}{c(1+\epsilon)^{2}},\dfrac{\epsilon^{2}}{c(1+\epsilon)^{2}})$&$ (0,0)$& $3+\succ+4\sim NB$ \\
				
				(6,3) &$(0,0)$&$(\dfrac{\epsilon}{c(1+\epsilon)^{2}},\dfrac{\epsilon^{2}}{c(1+\epsilon)^{2}})$&$(\dfrac{\epsilon}{c(1+\epsilon)^{2}},\dfrac{\epsilon^{2}}{c(1+\epsilon)^{2}})$& $3+\sim+4\succ NB$ \\
				
				(7,4) &$(0,0)$&$(\dfrac{\epsilon}{c(1+\epsilon)^{2}},\dfrac{\epsilon^{2}}{c(1+\epsilon)^{2}})$&$ (0,0)$& $3+\succ+4\sim NB$ \\
				
				(4,0), (5,1), (6,2) &$(0,0)$&$(0,0)$ &$ (0,0)$& $3+\sim+4\sim NB$ \\

				(7,3) &$(0,0)$&$(0,0)$ &$(\dfrac{\epsilon}{c(1+\epsilon)^{2}},\dfrac{\epsilon^{2}}{c(1+\epsilon)^{2}})$& $+4\succ3+\sim NB$ \\
				
				(5,0), (6,1), (7,2) &$(0,0)$&$(0,0)$ &$ (0,0)$& $3+\sim+4\sim NB$ \\
				
				(6,0), (7,1)&$(0,0)$&$(0,0)$ &$ (0,0)$& $3+\sim+4\sim NB$ \\
				
				(7,0)&$(0,0)$&$(0,0)$ &$ (0,0)$& $3+\sim+4\sim NB$ \\
				\hline
		\end{tabular}}
	\end{flushleft}
	\caption{Comparison of point systems}
\end{table}

\newpage
\section{The Dynamic Model}
In this section we model a rugby game following the argumentation line in  Mass\'o - Neme (1996). We conceive it as a dynamic game in which the feasible and equilibrium payoffs of the teams under the three point systems can be compared.
In this setting, we first find the minimax feasible payoffs in every point system.\footnote{The smallest payoff which the other team can force a team to receive. Formally: $\bar{v}_i = \min_{s_{-i}} \max_{s_i} \mathbf{U}(s_i, s_{-i})$.} This minimax payoff defines a region of equilibrium payoffs. We consider the Nash equilibriums used to reach this minimax payoffs in and take the average joint efforts in each system. Again, we want to find which point system makes the teams spend more effort in order to attain the equilibirum payoffs.\\ 
%More precisely, the idea to be analyzed here is that a point system that ensures a larger set of feasible outcomes and higher equilibrium payoffs induces  teams to make a higher effort, creating more alternatives in the game, and thus making it more entertaining. 
Formally, let us define a dynamic game as  $G=(\{A,B\}, (W,(0,0)),E^{\ast},T)$, where:
\begin{enumerate}
\item There are again two teams, $A$ and $B$. A generic team will be denoted by $i$.
\item We restrict the choices of actions to a  finite set of joint actions   $E^{\ast}$ where $E^{\ast}=\{(e_{A},e_{B})\in\mathbb{R}^{2}_+\}$ where each $e_{i}$ was used in a Nash equilibrium of the static game. 
\item A finite set of events $W$, each of which represents a class of equivalent pairs of scores of the two teams. %Depending on the point system  we have different characterizations of $W$ ($/$ is used to denote the quotient set):
\begin{itemize}
\item $(a,b)\sim(7,1)$ if $a>7$ and $b=1$\\
$(a,b)\sim(7,2)$ if $a>7$ and $b=2$\\
$(a,b)\sim(7,3)$ if $a>7$ and $b=3$\\
$(a,b)\sim(7,4)$ if $a-b\geq3$ and $b\geq4$\\
$(a,b)\sim(7,5)$ if $a-b\geq2$ and $b\geq5$\\
$(a,b)\sim(7,6)$ if $a-b\geq1$ and $b\geq6$\\
$(a,b)\sim(7,7)$ if $a=b$ and $b\geq7$\\
\end{itemize}

In each case, we say that two scores belong to the same event if the two teams get the same payoffs in both cases in a finite {\bf instantaneous game} in normal form defined as
 $$(a,b)=(\{A,B\},E^{\ast},((^{(a,b)}\!U_{i}^{S})_{i\in \{A,B\}} )$$
\noindent where $S = NB, +4$ or $3+$ and $^{(a,b)}\!U_{i}^{S}$ represents the utility function of team $i$ used in the static model in the instantaneous game in the event $(a,b)$, with the point system $S$. %This game requires the players to jointly choosing a Nash equilibrium when the scores are $(a,b)$. Notice that we assume such a game for every possible pair of scores.
\item All the point systems have the same initial event, namely $(0,0)$.
\item A transition function $T$, which specifies the new event as a function of the current event and the joint actions taken by both teams. 
Therefore
$$T:W\times E^{\ast}\longrightarrow W.$$ 
The transition function has only three possible outcomes (we use a representative element, i.e. a pair of scores, for any event in $W$):
\begin{enumerate}
\item $T((a,b),(e_{A},e_{B}))=(a,b)$
\item $T((a,b),(e_{A},e_{B}))=(a+1,b)$
\item $T((a,b),(e_{A},e_{B}))=(a,b+1)$
\end{enumerate}
These outcomes represent the fact that, upon a choice of joint efforts, either no team scores, $A$ scores or $B$ scores, respectively.  
\end{enumerate}  
Some further definitions will be useful in the rest of this work:

\begin{defi} For every $t\in \mathbb{N}$, define $H^{t}$ as $\overbrace{E^{\ast} \times \ldots \times E^{\ast}}^{t \ \mbox{times}}$ i.e. an element $h\in H^{t}$ is a history of joint efforts of length $t$. We denote by $H^{0}=\{e\}$ the set of histories of length $0$, with {\em e} standing for the empty history. 
Let $H=\cup^{\infty}_{t=0}H^{t}$ be the set of all possible histories in $G$. 
\end{defi}

We define recursively a sequence of $t +1$ steps of events starting with $(a,b)$, namely $\{(a,b)^j\}_{j=0}^t$, where $(a,b)^0 = (a,b)$, $\ldots$, $(a,b)^{t-1}= T((a,b), h_{t-1} \setminus h_{t-2})$,  $(a,b)^t=T((a,b)_{t-1}, h_{t} \setminus h_{t-1})$, where $(h_{0},\ldots,h_{t})\in H$ is such that for each $j=0, \ldots, t$,  $h_{j-1}$ is the initial segment of $h_j$ and $h_j \setminus h_{j-1}$ is the event exerted at the $j$-th step. 

\begin{defi}
A strategy of team $i\in\{A,B\}$ in the game $G$ is a function $f_{i}:H\rightarrow E^{\ast}_i$ such that for each $h_{t-1}$, the ensuing $h_t$ is the sequence $(h_{t-1}, (f_A(h_{t-1}), f_B(h_{t-1}))$.\footnote{This assumes {\em perfect monitoring}. That is, that teams decide their actions knowing all the previous actions in the play of the game.} We will denote by $F_{i}$ the set of all these functions for team $i$ and, by extension, we define $F=F_{A}\times F_{B}$.
\end{defi}

Thus, any $f\in F$ defines recursively a sequence of consecutive histories. We also have that $f$ defines a sequence of instantaneous games for each scoring system $S$ defined as $(a,b)^S(f)=\{(a,b)^{S}_{t}(f)\}_{t=0}^{\infty}$, where each game corresponds to an event $(a,b)^S_j$:  $(a,b)^S_{0}=(a,b)$ and for every $t\geq 1$, $(a,b)^S_{t} =T((a,b)^S_{t-1}, f(h_t))$. 

\begin{defi} A joint strategy $f = (f_A, f_B) \in F$ is {\em stationary} if for every $h,h^{\prime}\in H$ such $h = h_t$ and $h^{\prime} = h_{t^{\prime}}$ and the event generated by both is the same, namely $(\bar{a},\bar{b})$, we have that $f(h)=f(h^{\prime})$.
\end{defi}

That is, a stationary strategy only depends on the event at which it is applied, and not on how the event was reached.\\
We note the set of stationary strategies as $\mathcal{S} \subseteq F$ and in what follows we only consider strategies drawn from this set. In other words, we assume that teams act disregarding how a stage of the game was reached and play only according to the current state of affairs. For example, if the match at $t$ is tied $(3,3)$, teams $A$ and $B$ will play in the same way, irrespectively of whether the  score before was $(3,0)$ or $(0,3)$.\\ 
It can be argued that the assumption of stationarity does not seem to hold in some real-world cases since the way a given  score is reached may take an emotional toll on teams. If, say,  $A$ is winning at $(4,0)$, and suddenly the score becomes $(4,4)$, the evidence shows that $A$'s players will feel disappointed and anxious, changing the incentives under which they act (Cresswell \& Eklund, 2006).\\
But the theoretical assumption of stationarity is clearly applicable to the case of matches between high performance teams. For  instance, consider the first round of the Rugby World Cup 2015, when All Blacks (New Zealand's national team), the best team of the world, played against Los Pumas (Argentina's team), an irregular team. At the start of the second half Los Pumas were 4 points ahead. All Blacks, arguably the best rugby team of the world (and one of the best in any sport (Conkey, 2017)), instead of losing temper kept playing in a ``relaxed'' mode. This ensured that they ended winning the game by 26 to 16 (Cleary, 2015). So, the assumption of stationarity seems acceptable for high performance teams, reflecting their mental strength.\\
We have that,
\begin{lemi}[Mass\'o-Neme (1996)]
Let $s\in \mathcal{S}$. There exist two natural numbers $M,R\in \mathbb{N}$ such that $(a,b)^{t+R}(s)=(a,b)^{t+R+M}(s)$ for every $t\geq 1$. That is, a stationary strategy ({\em every} strategy in our framework)  produces a finite cycle of instantaneous games of length $M$ after $R$ periods.
\end{lemi} 
For every $s\in \mathcal{S}$ and a scoring system $S$ we define $b^S_{(a,b)}(s)=\{(a,b)^S_{1}(s),\ldots,$  $\ldots (a,b)^S_{R}(s)\}$ and $c^S_{(a,b)}(s)=\{(a,b)^S_{R+1}(s),\ldots, (a,b)^S_{R+M}(s)\}$ as the initial path and the cycle of instantaneous games generated by $s$, where $R$ and $M$ are the smallest numbers of Lemma 1. There are many ways in which the games in a cycle can be reached from the outcomes of another one:
\begin{defi}
Consider $s^{l},s^{{l}^{\prime}}\in \mathcal{S}$ and $(a,b)$ an initial event under a point system $S$. 
\begin{enumerate}
\item We say that $s^{l}$ and $s^{{l}^{\prime}}$ are directly connected, denoted $s^{l}\sim s^{{l}^{\prime}}$, if \mbox{$c_{(a,b)}^S (s^{l})\cap c_{(a,b)}^S(s^{{l}^{\prime}})\neq \emptyset$}.
\item We say that $s^{l}$ and $s^{{l}^{\prime}}$ are connected, denoted $s^{l}\approx s^{{l}^{\prime}}$, if there exist \mbox{$s^{1},\ldots, s^{m}\in \mathcal{S}$} such that $s^{l}\sim s^{1}\sim \ldots \sim s^{m}\sim s^{{l}^{\prime}}$.\footnote{In words, two strategies are directly connected if from the cycle of instaneous games corresponding to one of them, teams have direct access to the cycle of instantaneous games of the other and vice versa. If instead, they are (not directly) connected, teams can access from one of the cycles to the instantaneous games of the other one through a sequence of stationary strategies.}
\end{enumerate}
\end{defi} 
Then, we have (for simplicity we assume an initial event $(a,b)$ and a scoring system $S$):
\begin{defi}
For every $i\in\{A,B\}$, $\mathbf{U}_{i}(s)=(\frac{1}{|s|})\Sigma_{r=1}^{M}U_{i}^{j(r)}((e_{A},e_{B})^{R+r}(s))$, where $U^{j(r)}_i$ is $i$'s payoff function in the instantaneous game  $(a,b)^{R+r}(s)$ and $(e_{A},e_{B})^{R+r}$ is the profile of choices in that game.
\end{defi}
This means that the payoff of a stationary strategy is obtained as the average of the payoffs of the cycle. To apply this result in our setting, we have to characterize the set of feasible payoffs of $G$:
\begin{defi}
A vector $v\in\mathbb{R}^{2}$ is {\em feasible} if there exists a strategy $s\in \mathcal{S}$ such that $v=(\mathbf{U}_A(s), \mathbf{U}_B(s))$.
\end{defi} 
We have that:
\begin{teo}[Mass\'o-Neme (1996)]
A vector $v\in\mathbb{R}^{2}$ is feasible if and only if there exists $\mathcal{S}(v)=\{s^{1},\ldots,s^{\overline{k}}\}\subseteq \mathcal{S}$ such that for every $s^{r}, s^{r'}\in \mathcal{S}(v)$, $s^{r}\approx s^{r'}$ and there exists $(\alpha^{1},\ldots,\alpha^{\overline{k}})\in \overline{\Delta}$ (the $\overline{k}$-dimensional unit simplex) such that
$$v=\Sigma_{k=1}^{\overline{k}}\alpha^{k} (\mathbf{U}_A(s^k), \mathbf{U}_B(s^k)).$$
\end{teo}
The definition of Nash equilibria in this game is the usual one:
\begin{defi}
A strategy $s^{\ast}\in \mathcal{S}$ is a Nash equilibrium of game $G$ if for all $i\in \{A,B\}$, $\mathbf{U}_{i}(s^{\ast})\geq \mathbf{U_{i}}(s^{\prime})$ for all $s^{\prime}  \in \mathcal{S}$, with $s^{\ast}_i \neq s^{\prime}_i$ while  $s^{\ast}_{-i} = s^{\prime}_{-i}$.
A vector $v\in\mathbb{R}^{2}$ is an equilibrium payoff of $G$ if there exists a Nash equilibrium of $G$, $s\in \mathcal{S}$, such that $(\mathbf{U}_A(s), \mathbf{U}_B(s)) = v$.
\end{defi}
The following result characterizes the equilibrium payoffs of $G$:
\begin{teo}[Mass\'o-Neme (1996)]
Let $v$ be a feasible payoff of $G$. Then $v$ is an equilibrium payoff if and only if there exist $s^{1},s^{2},s^{3}\in \mathcal{S}$, and $(\alpha^{1},\alpha^{2},\alpha^{3})\in \Delta^{3}$ such that $v=\Sigma_{k=1}^{3}\alpha^{k}(\mathbf{U}_A(s^k), \mathbf{U}_B(s^k))$ and the payoff $v_{i}$ is better or equal than the higher payoff that team $i$ can guarantee by itself through a deviation of the cycles of $s^{1},s^{2},s^{3}$, the connected cycles and the initial path from those strategies.
\end{teo}

In words: $v$ is an equilibrium payoff if each team gets a better pay than the pay they can get if they deviate.\\
Finally, with all these definitions and theorems at hand, we can analyze the game where two teams that are far away in the position table face each other, so we set $\epsilon=0$ in the utility functions, disregarding the importance of blocking the attacks of the other team and focusing on scoring tries.%\begin{itemize}
%\item $T((a,b),(e_{A},e_{B}))= \left\{ \begin{array}{lcc}
%             (a+1,b) &   if  & |e_{A}-e_{B}|\leq0.5
 %            \\(a,b) & if &  0.5<|e_{A}-e_{B}|\leq0.75
  %           \\ (a,b+1) &  if & |e_{A}-e_{B}|>0.75
            
   %          \end{array}
  % \right.$

%\end{itemize} 
%This function gives a clear advantage to team $A$, which can be understood as the local team. We also consider utility functions with $\epsilon=0$. This can be seen as a match between two teams that are far away in the tournament, so they only care to maximize their payoffs, and do not care how much the other team wins.\\
The set of feasible payoffs of this game is given, as said, by a convex combination of payoffs of stationary strategies. The feasible payoffs region corresponding to each point system are represented in Figures 2 to 4. In each figure we can also see the minimax payoffs for the cycles favoring team $A$.\footnote{Where $c(a,b)$ is the cycle $\{(a,b)\}$ and $c(a,b)(c,d)$ is the cycle $\{(a,b),(c,d)\}$.} Every feasible payoff above and to the right of the minimax payoff is an equilibrium payoff. \\
Figure 2 shows the results for the $NB$ system and Table 5 shows the cycles that receive that minimax payoffs. 
\begin{figure}[H]
	\includegraphics[scale=1.2]{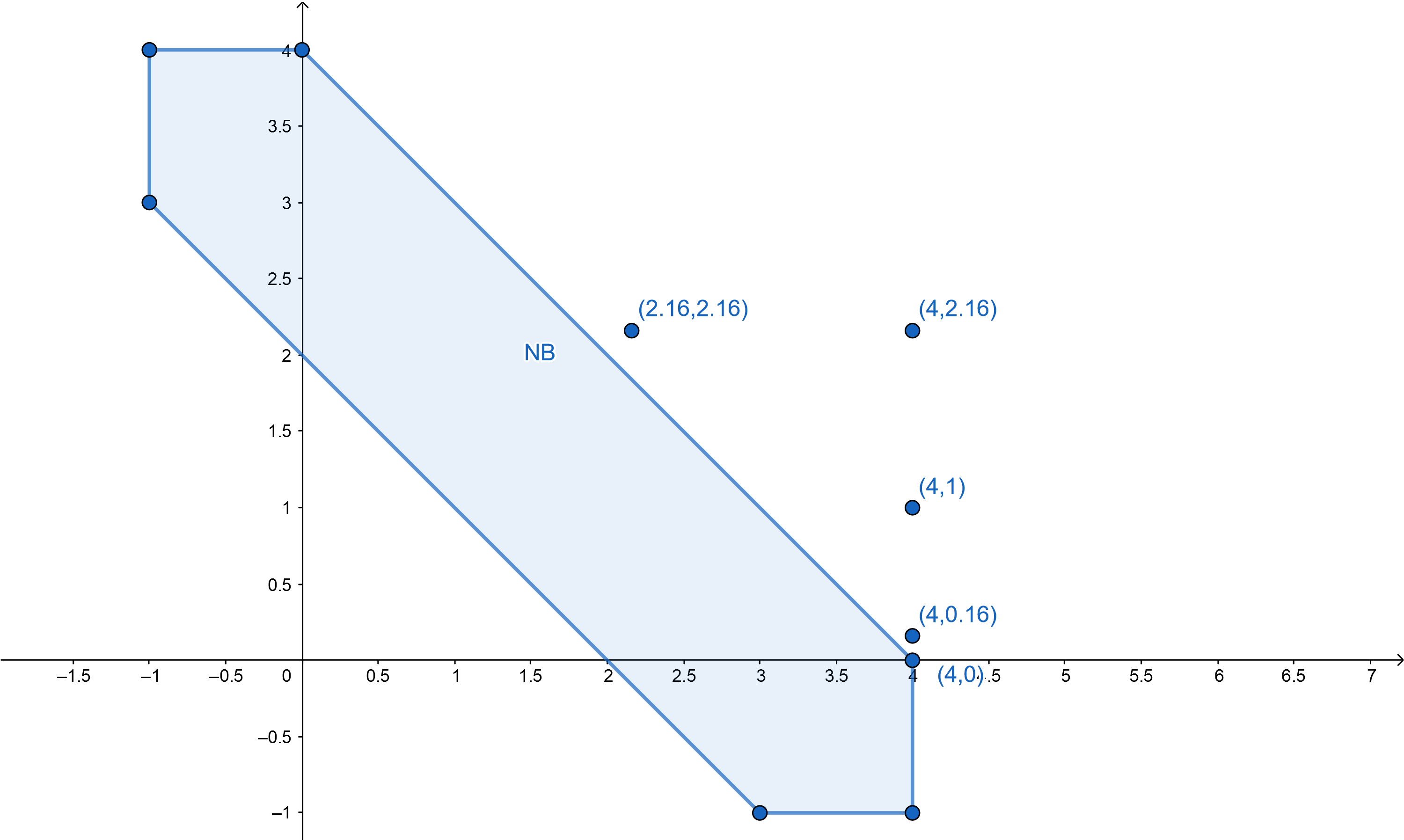}
	\caption{Feasible and Minimax payoffs in $NB$}
\end{figure} 
\begin{table}[H]
\centering
\begin{tabular}{|c |c|}
\hline
Minimax payoff & Cycle\\
\hline
$(2.16,2.16)$ & $c(0,0),c(1,1),c(2,2),c(3,3),c(4,4),c(5,5),c(6,6),c(7,7)$\\
$(4,1)$ & $c(1,0),c(2,1),c(3,2),c(4,3),c(5,4),c(5,5),c(7,6),c(7,5)(7,6)$\\
$(4,0.16)$ & $c(2,0),c(3,1),c(4,2),c(5,3),c(6,4),c(7,5),c(7,4)(7,5)$\\
$(4,0)$ & $c(3,0),c(4,1),c(5,2),c(6,3),c(7,4)$\\
&$c(4,0),c(5,1),c(6,2),c(7,3),c(5,0),c(6,1),c(7,2)c(6,0),c(7,1),c(7,0)$\\
$(4,2.16)$ & $c(7,6)(7,7)$\\
\hline
\end{tabular}
\caption{NB System}
\end{table}
Figures 3 and Table 6 show the results for the $3+$ system.
\begin{figure}[H]
  \includegraphics[scale=1.2]{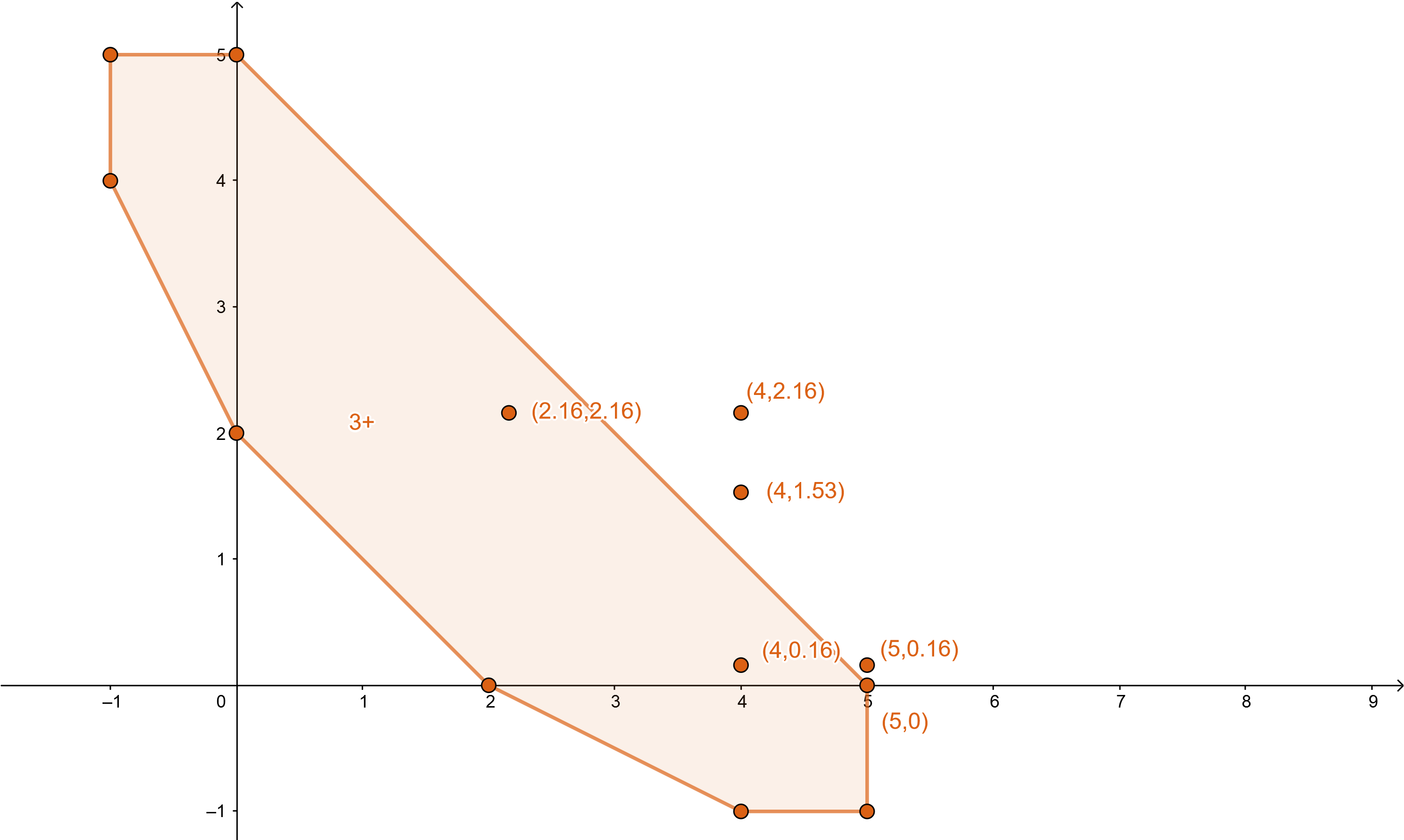}
  \caption{Feasible payoffs and Minimax payoffs in $3+$}
\end{figure}
\begin{table}[H]
\centering
\begin{tabular}{|c |c|}
\hline
Minimax payoff & Cycle\\
\hline
$(2.16,2.16)$ & $c(0,0),c(1,1),c(2,2),c(3,3),c(4,4),c(5,5),c(6,6),c(7,7)$\\
$(4,1.53)$ & $c(1,0),c(2,1),c(3,2),c(4,3),c(5,4),c(5,5),c(7,6),c(7,5)(7,6)$\\
$(4,0.16)$ & $c(2,0),c(3,1),c(4,2),c(5,3),c(6,4),c(7,5)$\\
$(5,0)$ & $c(3,0),c(4,1),c(5,2),c(6,3),c(7,4)$\\
&$c(4,0),c(5,1),c(6,2),c(7,3),c(5,0),c(6,1),c(7,2)c(6,0),c(7,1),c(7,0)$\\
$(4,2.16)$ & $c(7,6)(7,7)$\\
$(5,0.16)$ & $c(7,4)(7,5)$\\
\hline
\end{tabular}
\caption{3+ System}
\end{table}
Finally, Figure 4 and Table 7 do the same for the $+4$ system.
\begin{figure}[H]
  \includegraphics[scale=1.2]{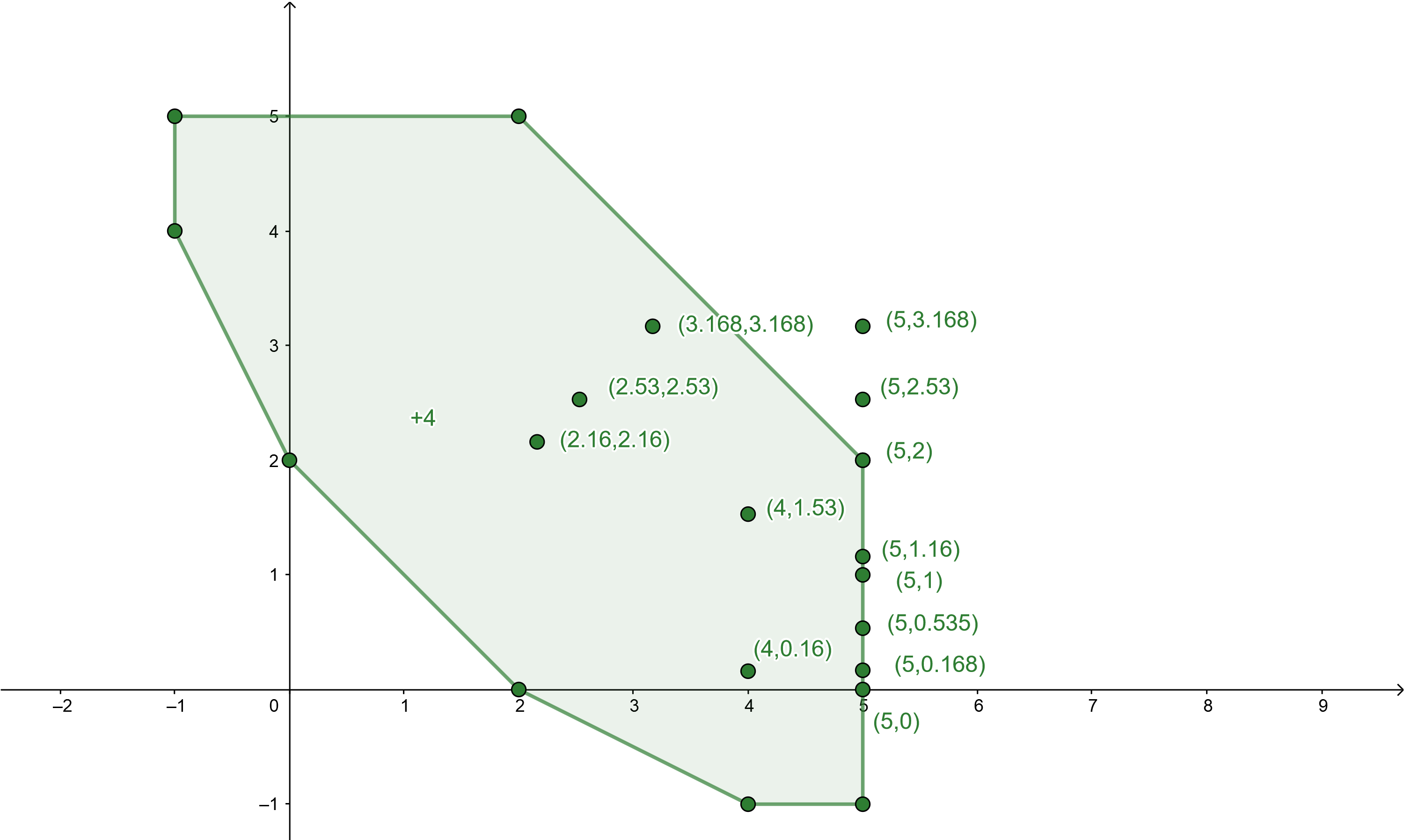}
  \caption{Feasible payoffs and Minimax payoffs in $+4$}
\end{figure} 
\begin{table}[H]
\centering
\begin{tabular}{|c |c|}
\hline
Minimax payoff & Cycle\\
\hline
$(2.16,2.16)$ & $c(0,0),c(1,1)$\\
$(2.53,2.53)$ & $c(2,2)$\\
$(3.168,3.168)$ & $c(3,3),c(4,4),c(5,5),c(6,6),c(7,7)$\\
$(4,1.53)$ & $c(1,0),c(2,1)$\\
$(5,2)$ & $c(3,2)$\\
$(5,2.53)$ &$c(4,3),c(5,4),c(6,5),c(7,6),c(7,5)(7,6)$\\
$(4,0.16)$ & $c(2,0)$\\
$(5,0.168)$ & $c(3,1),c(5,2)$\\
$(5,0.535)$ & $c(4,2)$\\
$(5,1.16)$ &$c(5,3),c(6,4),c(7,5),c(7,4)(7,5)$\\
$(5,0)$ &$c(3,0),c(4,1),c(4,0),c(5,1),c(6,2),c(5,0)$\\&$c(6,1),c(7,2),c(6,0),c(7,1),c(7,0)$\\
$(5,1)$ &$c(6,3),c(7,4),c(7,3)$\\
$(5,3.168)$ & $c(7,6)(7,7)$\\

\hline
\end{tabular}
\caption{+4 System}
\end{table}

The fact that some minimax payoffs are outside the feasible region indicates that some cycles do not have equilibrium payoffs, so one or both of the teams have incentives to change strategies and get a better payoff.
When we consider the joint efforts that yield the minimax payoffs we obtain an average joint effort of $(0,0)$ in the $NB$ system, $(0.18,0.18)$ in the $3+$ system and $(0.1776,0.1776)$ in the $+4$ system.

\newpage
\section{Empirical Evidence}

In order to check the empirical soundness of our theoretical analyses we will use a database of 473 rugby matches. They were played from 1987 to 2015 in different competitions, including the Rugby Word Cup, the Six Nations and club tournaments. We compiled this database drawing data from different sources ([12]-[23]). Each match is represented by a vector with four components, namely the number of tries of the local team, the number of tries of the visiting team, as well as the scores of the winning and the losing team, respectively.  \\
We perform a Least Squares analysis to explain the number of tries of each team and the differences in scores in terms of some explanatory variables. We consider as such the scoring system used in each match (our key variable), the nature of each team (a club team or a national team), a time trend and a constant. The selection of this kind of analysis is justified by, on one hand, its simplicity, but on the other because we lack a panel or temporal structure which could provide a richer information. Notice also that it is natural to posit a linear model in the presence of categorical variables (e.g. the scoring system in a tournament) (Wooldridge, 2020). \\
We run OLS regressions on different variants of the aforementioned general model, changing the way in which explanatory variables are included of changing the sample of matches to be analyzed. In the latter case we divided the entire sample in terms of the homogeneity or heterogeneity of teams playing in each match. In all cases we had to use robust errors estimators to handle the heteroskedasticity of the models. Also, assuming that each tournament is idiosyncratic, we controlled for clustered errors. \\
The general functional form of the model can be stated as:

\begin{equation}
\label{eqn:model}
T_i=\beta_0 + \beta_1 C2_i + \beta_2 C3_i + \gamma X_i + \epsilon_i
\end{equation}

There are many alternative ways of characterizing the dependent variable, which represents the number of tries in a match $i$, i.e. $T_i$. The first and obvious choice is to define it as the total number of tries in a match. But we also analyze variants in which we allow $T_i$ to represent either the number of tries of the local team, of the visitor team, the difference between them, those of the winning team (be it as a local or visiting team) and those of the losing team. \\
With respect of our variable of interest, i.e. the system of bonus points, we specify $+4$  as the categorical base, to compare it to the $3+$ and no bonus systems, represented by means of dummy variables, denoted $C2$ and $C3$ for $NB$ and $3+$ respectively. Both $\gamma$ and $X$ are vectors, containing the control variables and their parameters. We will vary the composition of $X$ in order to check the robustness of the effects of the scoring systems. Finally, $\beta_0$ is the constant, while $\epsilon$ is the error term (specified to account for heterokedasticity or clustered errors) \\
We will first present the descriptive statistics of the database. Then we give the results of the regressions on the different models built by varying both the definition of the dependent and the explanatory variables. Finally, we divide the sample in the classes of matches played by homogeneous or heterogeneous rivals, to compare their results for the same model. 
\subsection{Descriptive Statistics}

Figure~\ref{tries} illustrates different aspects of the distribution of the number of tries in the database of matches. Notice that the number of matches is not the same under the three scoring methods: for $+4$ we have $260$, $93$ under $NB$ and $120$ under $3+$. Nevertheless the evidence indicates that the $3+$ scoring method yields the highest scores, hinting that it is the one that induces a more aggressive play. \\
\begin{figure}[hbt!]
	\centering
	\begin{subfigure}[t]{0.75\textwidth}
		\centering
		\scriptsize
		\includegraphics[width=1\textwidth]{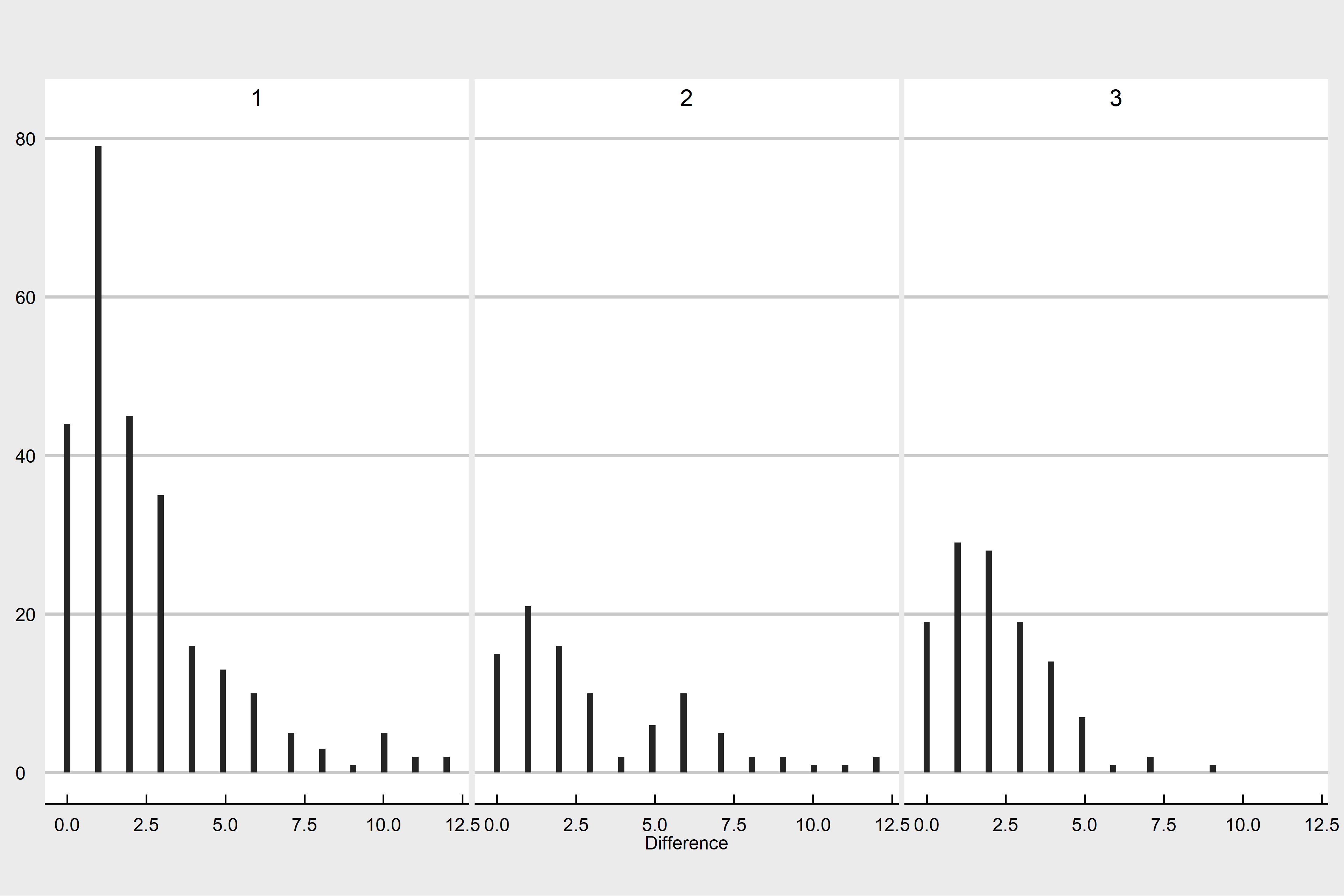}
		\caption{Differences of tries \label{fig:diff}}

	\end{subfigure}%
	 
	\begin{subfigure}[t]{0.75\textwidth}
		\centering
		\scriptsize
		\includegraphics[width=1\textwidth]{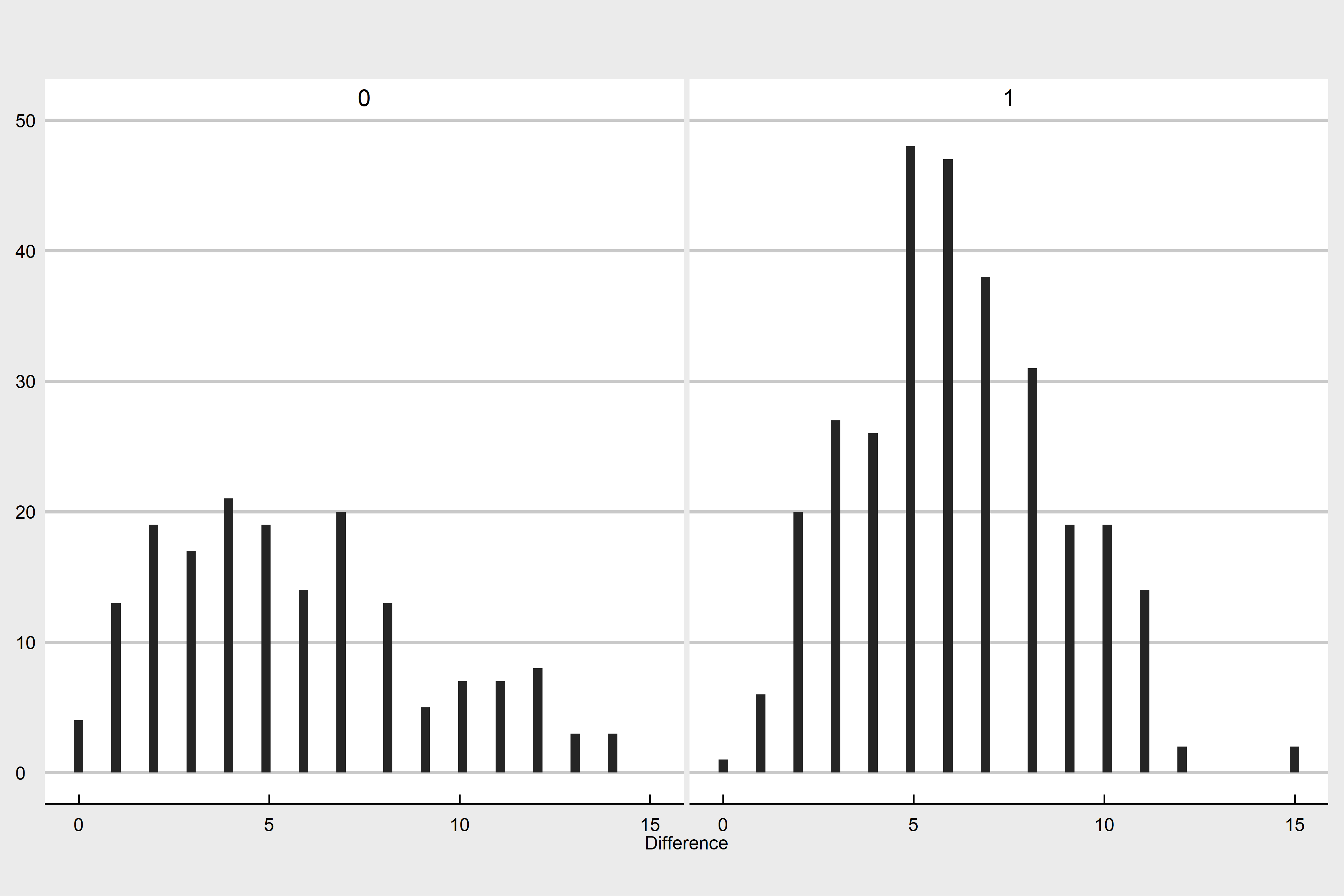}
		\caption{Differences for club teams \label{fig:diffclub}}
	\end{subfigure}

\caption{Histograms of distributions of differences of tries in each match.}
	
\label{diff}
\end{figure}
\begin{figure}[hbt!]
	\centering	
	\begin{subfigure}[t]{0.75\textwidth}
		\centering
		\scriptsize
		\includegraphics[width=1\textwidth]{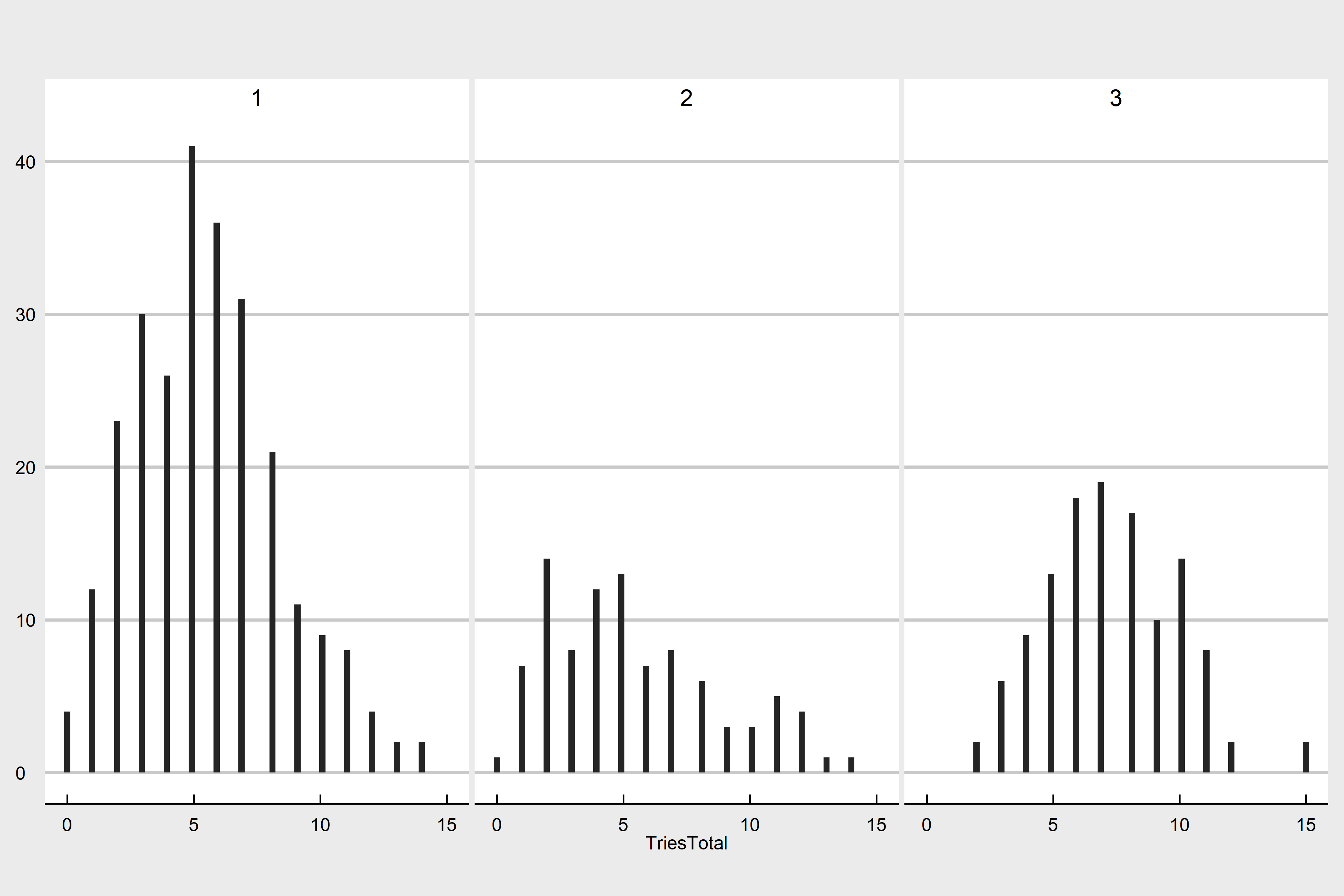}
		\caption{Total tries \label{fig:TriesTotal}}
	\end{subfigure}
	 	
	\begin{subfigure}[t]{0.75\textwidth}
		\centering
		\scriptsize
		\includegraphics[width=1\textwidth]{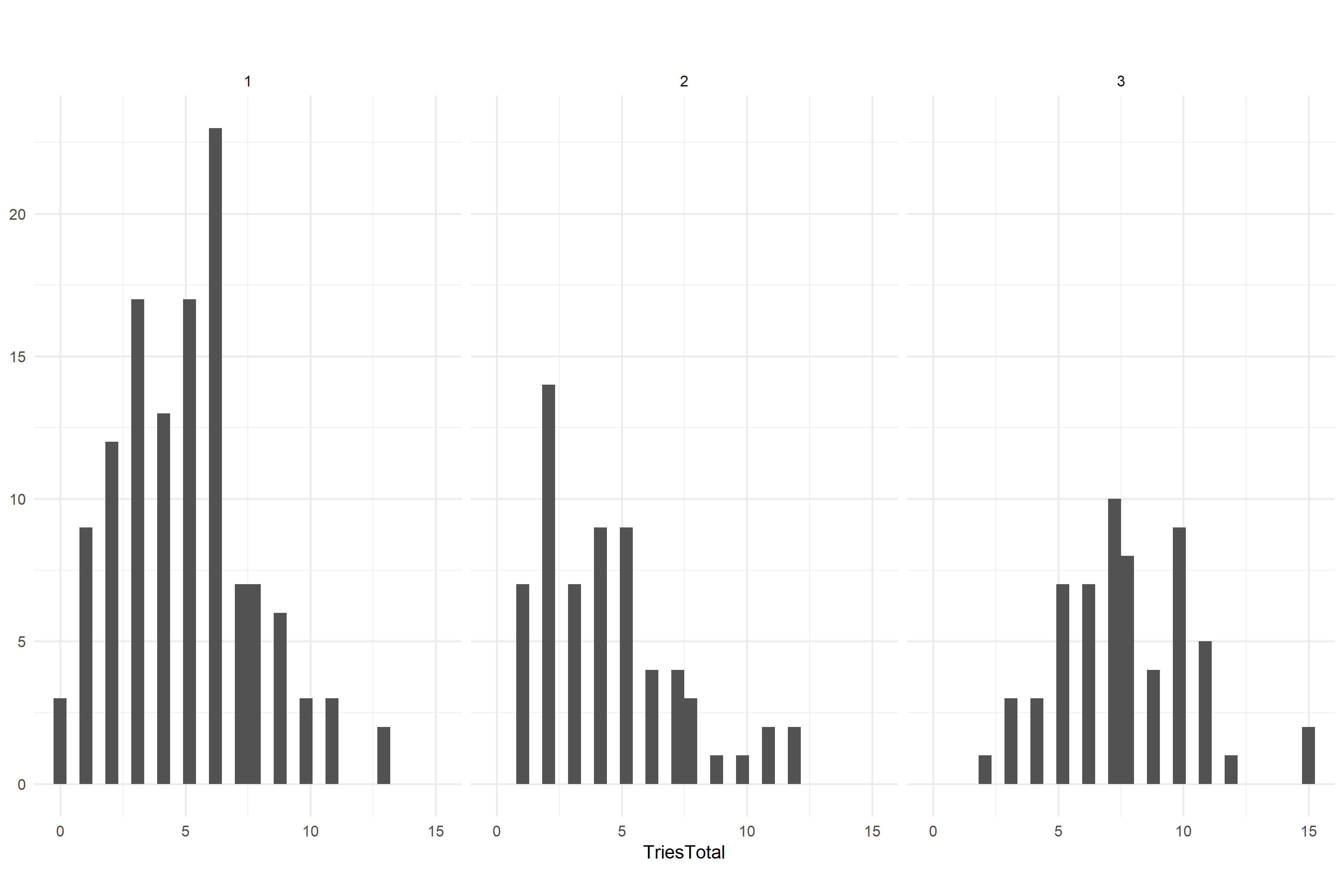}
		\caption{Total tries in homogeneous matches \label{fig:TriesTotalh}}
	\end{subfigure}
	
	\caption{Histograms of distributions of tries.}	
\label{tries}	
\end{figure}

\subsection{Samples and Exploratory Regressions}

We run regressions on different specifications of the general model represented by expression \eqref{eqn:model}  in order to make inferences beyond the casual evidence. We use the variable $code$, to represent the scoring  system, with the base value $1$ for $+4$, $2$ for $NB$ and $3$ for $3+$, as we expressed above with the variables $C2$ and $C3$ in \eqref{eqn:model}. For $T_i$ we use different specifications, namely $TriesTotal$, $TriesLocal$ and $TriesVis$, representing the number of total tries, tries by the local team and tries of the visiting team, respectively. With respect to the control variables $X$ we use different selections from a set that includes $SR$ is a dummy variable indicating that a match correspond to a Super Rugby tournament ( because Super Rugby is clearly different from the other tournaments analyzed here); $Club$, which indicates whether a match is played by club teams or not; $previous$, a dichotomous variable taking value $1$ on the older matches in our database, namely those played between 1987 and 1991. Of particular interest are two variables that can be included in $X$. One is  $diff= |TriesLocal - TriesVis|$ representing the difference in absolute value between the tries of the local and the visiting team. The other, related, control variable is $diff2 = TriesLocal - TriesVis$, capturing the possible advantage of being the local team. Finally, we include $year$ as to capture the possible existence of a temporal trend. \\ 
The results can be seen in Table~\ref{table}.\footnote{All tables of this section can be found at the end of the article.} It can be seen that $3+$ is indeed the scoring method that achieves the highest number of tries, namely between $1$ and over $2$ more than $+4$ (which is our benchmark). $NB$ induces, in general, less tries than $+4$, except in the case of number of tries of the visiting teams. \\ 
With respect to the control variables, we can see that $SR$ has a negative impact while $Club$ and $previous$ have a positive influence. The time trend is not significant in any of the regressions. \\
\begin{sidewaystable}[hbt!]
	\centering
	\scriptsize
	\caption{General Case}
\label{table}
\begin{tabular}{lcccccccccccc} \hline
	& (1) & (2) & (3) & (4) & (5) & (6) & (7) & (8) & (9) & (10) & (11) & (12) \\
	VARIABLES & TriesLocal & TriesVis & TriesTotal & diff & diff2 & TriesLocal & TriesVis & TriesLocal & TriesVis & TriesLocal & TriesVis & TriesTotal \\ \hline
	&  &  &  &  &  &  &  &  &  &  &  &  \\
	TriesVis &  &  &  &  &  & -0.146** &  &  &  &  &  &  \\
	&  &  &  &  &  & (0.067) &  &  &  &  &  &  \\
	2.code & -2.339*** & 0.760*** & -1.579*** & -1.576*** & -3.099*** & -2.228*** & 0.624** &  &  & -1.514*** & 0.760*** & 0.799 \\
	& (0.522) & (0.284) & (0.585) & (0.488) & (0.603) & (0.521) & (0.290) &  &  & (0.395) & (0.284) & (0.590) \\
	3.code & 1.475*** & 0.975*** & 2.450*** & 0.339* & 0.006 & 1.617*** & 1.061*** & 1.475*** & 0.975*** &  &  & 2.159*** \\
	& (0.292) & (0.250) & (0.369) & (0.197) & (0.303) & (0.298) & (0.255) & (0.291) & (0.250) &  &  & (0.466) \\
	Club & -1.375*** & 1.088*** & -0.287 & -1.843*** & -2.710*** & -1.216*** & 1.008*** &  &  &  & 1.088*** &  \\
	& (0.425) & (0.180) & (0.461) & (0.398) & (0.449) & (0.428) & (0.184) &  &  &  & (0.180) &  \\
	previous & 2.022*** & 0.103 & 2.125*** & 2.056*** & 1.919** & 2.037*** & 0.220 &  &  & 2.022*** & 0.103 &  \\
	& (0.632) & (0.401) & (0.654) & (0.575) & (0.832) & (0.621) & (0.410) &  &  & (0.631) & (0.401) &  \\
	SR & -0.775*** & -0.033 & -0.808** &  &  & -0.780*** & -0.078 &  &  &  &  &  \\
	& (0.278) & (0.246) & (0.367) &  &  & (0.280) & (0.247) &  &  &  &  &  \\
	TriesLocal &  &  &  &  &  &  & -0.058** &  &  &  &  &  \\
	&  &  &  &  &  &  & (0.027) &  &  &  &  &  \\
	Club &  &  &  &  &  &  &  &  &  &  &  & 1.883*** \\
	&  &  &  &  &  &  &  &  &  &  &  & (0.523) \\
	year &  &  &  &  &  &  &  &  &  &  &  & -0.004 \\
	&  &  &  &  &  &  &  &  &  &  &  & (0.031) \\
	Constant & 4.650*** & 1.262*** & 5.912*** & 3.688*** & 3.388*** & 4.834*** & 1.533*** & 2.500*** & 2.317*** & 3.825*** & 1.262*** & 10.969 \\
	& (0.391) & (0.111) & (0.398) & (0.380) & (0.414) & (0.407) & (0.169) & (0.222) & (0.200) & (0.192) & (0.111) & (63.039) \\
	&  &  &  &  &  &  &  &  &  &  &  &  \\
	Observations & 473 & 473 & 473 & 473 & 473 & 473 & 473 & 180 & 180 & 293 & 293 & 245 \\
	R-squared & 0.090 & 0.152 & 0.102 & 0.123 & 0.110 & 0.098 & 0.159 & 0.113 & 0.077 & 0.045 & 0.076 & 0.203 \\ \hline
	\multicolumn{13}{c}{ Robust standard errors in parentheses} \\
	\multicolumn{13}{c}{ *** p$<$0.01, ** p$<$0.05, * p$<$0.1} \\
\end{tabular}
\end{sidewaystable}
\subsection{The Homogeneous Case}

Table~\ref{table2} presents the results of running the aforementioned regressions but only on the class of matches between homogeneous teams.\footnote{National teams are considered homogeneous if they are in the same Tier ([24]), and clubs are considered homogeneous if they belong to the same country.} The dependent variables of the regressions are on the first rows, where the first four columns indicate robust errors while the other four give the errors clustered by tournaments. \\
The transition from  $+4$ to $NB$ does not make a difference in robust errors but it does so for tournament errors, adding a little more than half a try (not for the losing team, for which it does not make any difference). The effect of changing from $+4$ to $3+$ is stronger, adding more than $2$ total tries and more than $1$ for the winning team. \\ 
On the other hand, any of the scoring systems induces almost $2$ more total tries in club tournaments than with national teams. Finally, nor $year$ or the $constant$ are significant. 

\begin{sidewaystable}[hbt!]
	\centering
	\caption{Homogeneus case}
\label{table2}
\begin{tabular}{lccccccccc} \hline
	& (1) & (2) & (3) & (4) & (5) & (6) & (7) & (8) & (9) \\
	VARIABLES & TriesTotal & TriesWin & TriesLoss & Diff & TriesTotal & TriesWin & TriesLoss & Diff & TriesTotal \\ \hline
	&  &  &  &  &  &  &  &  &  \\
	2.code & 0.799 & 0.731* & 0.067 & 0.664* & 0.799*** & 0.731*** & 0.067 & 0.664*** & 0.799 \\
	& (0.590) & (0.438) & (0.224) & (0.369) & (0.203) & (0.159) & (0.125) & (0.201) & (0.590) \\
	3.code & 2.159*** & 1.364*** & 0.873*** & 0.491* & 2.159*** & 1.364*** & 0.873*** & 0.491** & 2.159*** \\
	& (0.466) & (0.319) & (0.211) & (0.289) & (0.538) & (0.304) & (0.213) & (0.180) & (0.466) \\
	Club & 1.883*** & 1.171*** & 0.635*** & 0.536* & 1.883*** & 1.171*** & 0.635** & 0.536** & 1.883*** \\
	& (0.523) & (0.367) & (0.213) & (0.306) & (0.542) & (0.308) & (0.223) & (0.215) & (0.523) \\
	year & -0.004 & -0.013 & 0.009 & -0.023 & -0.004 & -0.013 & 0.009 & -0.023 & -0.004 \\
	& (0.031) & (0.029) & (0.011) & (0.031) & (0.031) & (0.023) & (0.009) & (0.016) & (0.031) \\
	Constant & 10.969 & 29.152 & -17.653 & 46.805 & 10.969 & 29.152 & -17.653 & 46.805 & 10.969 \\
	& (63.039) & (58.695) & (21.625) & (62.070) & (61.675) & (45.887) & (17.956) & (32.588) & (63.039) \\
	&  &  &  &  &  &  &  &  &  \\
	Observations & 245 & 245 & 245 & 245 & 245 & 245 & 245 & 245 & 245 \\
	R-squared & 0.203 & 0.138 & 0.213 & 0.030 & 0.203 & 0.138 & 0.213 & 0.030 & 0.203 \\ \hline
	\multicolumn{10}{c}{ Robust standard errors in parentheses} \\
	\multicolumn{10}{c}{ *** p$<$0.01, ** p$<$0.05, * p$<$0.1} \\
\end{tabular}
\end{sidewaystable}

\subsection{The Non-Homogeneous Case}

This analysis, represented in Table~\ref{table3} is performed on the same variables and with the same interpretation of errors as the previous case, but including all the matches. \\
We do not find differences between $NB$ and $+4$. $3+$, instead, makes a difference, although with a lower impact than in the homogeneous case.  Another relevant difference is that in this case the effect of $Club$ gets reversed. That is, winning teams score less while losing ones more, reducing in almost $2$ tries the differences with national teams. \\
Another interesting feature is that $year$ becomes significant. That is, there exists a trend towards increasing the differences in time. \\

\begin{sidewaystable}[hbt!]
	\centering
	\caption{Non Homogeneus case}
	\label{table3}
\begin{tabular}{lccccccccc} \hline
	& (1) & (2) & (3) & (4) & (5) & (6) & (7) & (8) & (9) \\
	VARIABLES & TriesTotal & TriesWin & TriesLoss & Diff & TriesTotal & TriesWin & TriesLoss & Diff & TriesTotal \\ \hline
	&  &  &  &  &  &  &  &  &  \\
	2.code & -0.440 & -0.463 & 0.023 & -0.486 & -0.440 & -0.463 & 0.023 & -0.486 & 0.799 \\
	& (0.526) & (0.487) & (0.147) & (0.490) & (0.712) & (0.683) & (0.080) & (0.663) & (0.590) \\
	3.code & 1.907*** & 1.122*** & 0.785*** & 0.336* & 1.907*** & 1.122*** & 0.785*** & 0.336*** & 2.159*** \\
	& (0.297) & (0.205) & (0.146) & (0.197) & (0.229) & (0.118) & (0.122) & (0.071) & (0.466) \\
	Club & -0.559 & -1.200*** & 0.641*** & -1.841*** & -0.559* & -1.200*** & 0.641*** & -1.841*** & 1.883*** \\
	& (0.438) & (0.395) & (0.138) & (0.398) & (0.253) & (0.193) & (0.131) & (0.211) & (0.523) \\
	year & 0.003*** & 0.002*** & 0.001*** & 0.002*** & 0.003*** & 0.002*** & 0.001*** & 0.002*** & -0.004 \\
	& (0.000) & (0.000) & (0.000) & (0.000) & (0.000) & (0.000) & (0.000) & (0.000) & (0.031) \\
	Constant &  &  &  &  &  &  &  &  & 10.969 \\
	&  &  &  &  &  &  &  &  & (63.039) \\
	&  &  &  &  &  &  &  &  &  \\
	Observations & 473 & 473 & 473 & 473 & 473 & 473 & 473 & 473 & 245 \\
	R-squared & 0.814 & 0.766 & 0.719 & 0.557 & 0.814 & 0.766 & 0.719 & 0.557 & 0.203 \\ \hline
	\multicolumn{10}{c}{ Robust standard errors in parentheses} \\
	\multicolumn{10}{c}{ *** p$<$0.01, ** p$<$0.05, * p$<$0.1} \\
\end{tabular}
\end{sidewaystable}

\subsection{Final Remarks}
All the results obtained, both in the general case and distinguishing between homogeneous and heterogeneous teams, indicate that the results of our theoretical models seem to hold in the real world. \\
%A rather surprising realization is that {\em NB} is better than {\em 3+} in Condorcet comparisons.  We can understand this by noticing that only two kind of events make the teams more offensive under the {\em 3+} system. Namely, when teams are tied and the utility of breaking the tie induces an increase in offensiveness; and when the difference in scores is of two tries and scoring one more try improves the utility of the leading team. In all other cases the teams have incentives to become more defensive, instead of offensive.\\ 
%An example illustrates this. Consider a game in which the score is $(3,0)$. 
%Team $B$ has nothing to lose by becoming more offensive, but team $A$ can lose its bonus point earned by winning for a difference of 3 tries. So, $A$ has the incentives to become highly defensive in order to keep its advantage.\\ 
%It is not so surprising that {\em +4} ends up ranking higher than {\em NB}. Giving a bonus point to the teams that score $4$ or more tries makes the teams more offensive. On the other hand, it is not obvious at first that {\em +4} should be better than {\em 3+}. But it can be understood in terms of the extra certainty ensured by {\em+4} of gaining a bonus point without the risk of losing it, tends to make teams more offensive.\\
%So, in conclusion, giving  an extra point for scoring $4$ or more tries and not taking it away seems to be the most effective way of making the teams offensive. We will examine whether the Condorcet ordering of the systems remains unchanged in a dynamic representation of rugby games.\\

\section{Conclusions}

The results of analyzing rugby games in theoretical and empirical terms are consistent. The $3+$ system induces  teams to exert more effort both in the static and empirical models. Moreover, in the particular instance analyzed of the dynamic model, with $\epsilon=0$, the result is the same. In all the models, we find that the $3+$ system ranks first, $+4$ second and $NB$ third.\\
While choosing different values of $\epsilon$ in the dynamical model may change a bit the results, it seems that a sports planner should use the $3+$ bonus point system if the goal is to make the game more entertaining.\\
Some possible extensions seem appropiate topics for future research. If we consider $\epsilon$ as a measure of the ``distance'' between teams playing in a league,  the choice of the appropriate bonus point system may depend on the teams and the moment of the tournament they are playing. Incentives at the beggining are not the same at the end of the tournament. The idea of conditionalizing the design of a tournament taking into account this in an optimal way, can be of high interest. %The results of analyzing static and dynamic representations of rugby games are, in some sense consistent, even if the models differ greatly. In both cases we find that the $+4$ system makes teams more offensive and creates a larger variety of feasible payoffs. Moreover, at least in the particular instance analyzed, it yields higher equilibrium payoffs.\\
%This is not true for the other two systems. While in the static case $NB$ makes the teams more offensive, the feasible payoffs regions are larger, in the dynamic version, with the $3+$ system. In the latter case we have that $3+$ yields higher equilibrium payoffs than $NB$.\\
%In any case, independently of whether we compare equilibria or feasible payoffs, the $+4$ system must be chosen if we seek to make the game more \textquotedblleft entertaining\textquotedblright.\\
%Some possible extensions seem appropriate topics for future research. An immediate option is carry out a general analysis of equilibrium payoffs, considering other specifications of $\alpha, \beta$ and $T$. The intuitions gained in this example seem to indicate that $+4$ will be still the best of the three systems.\\
%We can also analyze the strategies of the dynamic model, and compare the Subgame Perfect Equilibria under the three point systems and look for the Condorcet winner among them. Finally, in a dual exercise it will be of high interest to, by comparing equilibria, find the {\em optimal} point system.

\end{document}